\documentclass[amsthm]{elsart}

\makeatletter
\newif\if@restonecol
\makeatother

\usepackage{yjsco}
\usepackage{natbib}
\usepackage{amsfonts}
\usepackage{amsmath}
\usepackage{amssymb}
\usepackage[algoruled,vlined,linesnumbered]{algorithm2e}

\SetKw{KwAnd}{and} \SetKw{KwFrom}{from} \SetKw{KwBreak}{break}
\SetKw{KwGlobal}{global} \SetKw{KwEnd}{end}

\begin{document}

\begin{frontmatter}

\title{Extractions: Computable and Visible Analogues of Localizations for Polynomial Ideals}

\author{Ye Liang}
\address{Max-Plank-Institut f\"ur Informatik, 66123 Saarbr\"ucken, Germany}

\ead{wolf39150422@gmail.com}

\begin{abstract}
When studying local properties of a polynomial ideal, one usually needs a theoretic technique called localization. For most cases, in spite of its importance, the computation in a localized ring cannot be algorithmically preformed. On the other hand, the standard basis method is very effective for the computation in a special kind of localized rings, but for a general semigroup order the geometry of the localization of a positive-dimensional ideal is difficult to interpret.

In this paper, we introduce a new ideal operation called extraction. For an ideal $I$ in a polynomial ring $K[x_1,\ldots,x_n]$ over a field $K$, we use another ideal $J$ to control the primary components of $I$ and the result $\beta(I,J)$ is called the extraction of $I$ by $J$. It is still a polynomial ideal and has a concrete geometric meaning in $\bar{K}^n$, i.e., we keep the branches of $\textbf{V}(I) \subset \bar{K}^n$ that intersect with $\textbf{V}(J) \subset \bar{K}^n$ and delete others, where $\bar{K}$ is the algebraic closure of $K$. This is what we mean by visible. On the other hand, we can use the standard basis method to compute a localized ideal corresponding to $\beta(I,J)$ without a complete primary decomposition, and can do further computation in the localized ring such as determining the membership problem of $\beta(I,J)$. Moreover, we prove that extractions are as powerful as localizations in the sense that for any multiplicatively closed subset $S$ of $K[x_1,\ldots,x_n]$ and any polynomial ideal $I$, there always exists a polynomial ideal $J$ such that $\beta(I,J)=(S^{-1}I)^c$.
\end{abstract}


\end{frontmatter}

\section{Introduction} \label{introduction}

Since localization was introduced to mathematics in the first half of the twentieth century, it has become an indispensable technique in commutative algebra and algebraic geometry. The basic philosophy behind it is simple: By making some elements invertible, several components of an ideal can be deleted and others will be kept, so that one can investigate the local properties of this ideal. Though localization is important in theoretic studies, there is no effective methods to compute it for nontrivial cases until a splendid piece of work of \cite{Mora82}.

The method of Mora relates to a notion called standard basis which was introduced by \cite{Hironaka64} and \cite{Buchberger65,Buch06} independently. Hironaka considered the local cases but did not provide algorithms. Buchberger presented a famous algorithm, i.e. the Buchberger algorithm, for global orders but there are no localizations in these cases. Mora provided the first algorithm to compute a standard basis for a local order. He borrowed the basic idea of the Buchberger algorithm and replaced the division process in the Buchberger algorithm by the so called Mora normal form algorithm. After theoretical and practical improvements by \cite{Lazard83}, \cite{Robbiano1985}, \cite{Graebe94}, \cite{Greuel96} and others, one can now effectively compute a standard basis for any semigroup order in a computer algebra software, for instance \textsc{Singular} designed by \cite{DGPS}. However, for a general semigroup order and a positive-dimensional polynomial ideal, the geometric meaning of the localization is not as clear as for a local or global order, though semigroup orders have several applications such as in the computation of Hilbert-Samuel functions in \cite{Mora95}, some local operations in \cite{Alonso91} and other stuffs.

In this paper, we introduce a new ideal operation in Definition \ref{Def:CPart} called extraction. Given two ideals $I$ and $J$ in a polynomial ring $A:=K[x_1,\ldots,x_n]$ where $K$ is a field, we can define another ideal in the following way. Take a minimal primary decomposition $I= \cap _{i=1}^kQ_i$ such that $Q_i+J \neq A$ for $i=1,\ldots,m$ and $Q_i+J=A$ for $i=m+1,\ldots,k$. Then, we call $\beta(I,J):=\cap _{i=1}^mQ_i$ the extraction of $I$ by $J$. This notion is well defined and has concrete geometric meaning in $\bar{K}^n$ where $\bar{K}$ is the algebraic closure of $K$, i.e., we only extract the components of $\textbf{V}(I) \subset \bar{K}^n$ that meet $\textbf{V}(J) \subset \bar{K}^n$ and delete other components. This process is somewhat similar to a localization but the operation is opposite.

To compute a localized ideal corresponding to the extraction $\beta(I,J)$ by the standard basis method, we need first to study the geometry of a special kind of semigroup orders, i.e. control orders in Definition \ref{Def:COrder}. As in \cite{Liang14} for zero-dimensional cases, for a control order $>$, the local variables determine an ideal $J:=\langle x_1,\ldots,x_r \rangle$. Then we claim in Theorem \ref{decomposition} that the contraction of $\textup{Loc}_>(I)$ is just $\beta(I,J)$. For a general $J$ we need a lifting by adding new variables to the polynomial ring and transform the general ideal $J$ to the simple case that we just deal with (cf. Theorem \ref{Th:Projection}).

Comparing to general localizations, the advantage of the notion extraction is that we can not only see the geometry in $\bar{K}^n$ but also can effectively and directly (without a complete primary decomposition of $I$, especially do not need to compute a Gr\"obner basis of $I$) compute a corresponding ideal to it in a localized ring by the standard basis method, so that we can study some properties such as the membership problem of $\beta(I,J)$ (cf. Corollary \ref{Corollary:membership} and Remark \ref{Remark:membership}). Moreover, we can prove in Theorem \ref{Th:powerful} that extractions are as powerful as general localizations in the sense that for a contraction $L$ of any localization of a polynomial ideal $I$, there always exists a polynomial ideal $J$ such that $L=\beta(I,J)$. When $I$ is zero-dimensional, we can even work out $\beta(I,J)$ from the localized ring w.r.t. a semigroup order in \cite{Liang14}. But if $I$ is positive-dimensional, then we have no general algorithm at the moment to compute the extraction directly from this localized ring, though we can always compute it by definition if we do not consider the efficiency.

The rest contents are structured as follows. In Section \ref{Sec:Pre}, we list some basic materials. Section \ref{Sec:Classify} is devoted to introducing a special kind of semigroup orders called control orders and to studying their geometric meanings. In Section \ref{Sec:ControlVarieties}, we introduce the notion extraction of a polynomial ideal by another ideal. Then, we study a relation on dimensions between control orders and contractions of ideals in localized rings in Section \ref{Sec:Dimension}, and finally study some basic properties of extractions in Section \ref{Sec:Basic}.

\newtheorem{theorem}{Theorem}[section]
\newtheorem{proposition}[theorem]{Proposition}
\newtheorem{lemma}[theorem]{Lemma}
\newtheorem{definition}[theorem]{Definition}
\newtheorem{observation}[theorem]{Observation}
\newtheorem{corollary}[theorem]{Corollary}
\newtheorem{problem}[theorem]{Problem}
\newtheorem{remark}[theorem]{Remark}
\newtheorem{example}[theorem]{Example}
\newtheorem{assumption}[theorem]{Assumption}

\section{Preliminaries}\label{Sec:Pre}
Let $K$ be a field, $\bar{K}$ be the algebraic closure of $K$, $A:=K[x_1,\ldots,x_n]$ and $\textup{T}^{\{x_1,\ldots,x_n\}}:=\{x^{\alpha}:\alpha \in \mathbb{Z}^n_{\geq 0} \}$.  A semigroup order $>$ on $\mathbb{Z}^n_{\geq 0}$, or on $\textup{T}^{\{x_1,\ldots,x_n\}}$ in $A$ is a total order compatible with multiplication of monomials. Local orders and global orders are semigroup orders satisfying the conditions that every variable is smaller than $1$ and every variable is larger than $1$, respectively. If a semigroup order is not local or global, then we call it a mixed semigroup order or a mixed order. It has at least one local variable and at least one global variable. Let $>$ be a semigroup order in $A$ and let $S=\{1+g \in A:g=0 \ \textup{or} \ \textsc{lt}(g)<1\}$ where $\textsc{lt}(g)$ is the leading term of $g$ w.r.t. $>$. The localization of $A$ w.r.t. $>$ is the ring $\textup{Loc}_{>}(A)=S^{-1}A=\{f/(1+g): f \in A, \ 1+g \in S\}$. Let $I \subset \textup{Loc}_{>}(A)$ be an ideal. A standard basis of $I$ w.r.t. $>$ is a set $\{g_1,\ldots,g_t\} \subset I$ such that $\langle \textsc{lt}(I) \rangle = \langle \textsc{lt}(g_1),\ldots,\textsc{lt}(g_t) \rangle$. Here $\textsc{lt}$ is a generalized version of the leading term function for nonzero elements in $\textup{Loc}_{>}(A)$. See \cite{Cox05} for the details.

Let $f: A \rightarrow S^{-1}A, \ a \mapsto a/1$ be a ring homomorphism where $S$ is a multiplicatively closed subset of $A$. For an ideal $I \subset A$, its extension $I^e$ in $S^{-1}A$ is $I(S^{-1}A)=S^{-1}I$. For an ideal $J \subset S^{-1}A$, its contraction $J^c$ in $A$ is $f^{-1}(J)$. The following theorem is basic and can be found in \cite{Atiyah1969}.
\begin{theorem}\label{Atiyah}
Let $S$ be a multiplicatively closed subset of $A$, and let $I$ be an ideal. Let $I=\cap^k_{i=1}Q_i$ be a minimal primary decomposition of $I$. Let $P_i$ be the radical of $Q_i$ and suppose the $Q_i$ numbered so that $S$ meets $P_{m+1},\ldots,P_k$ but not $P_1,\ldots,P_m$. Then $$ S^{-1}I=\cap^m_{i=1}S^{-1}Q_i, \quad I^{ec}=(S^{-1}I)^c=\cap^m_{i=1}Q_i$$ and these are minimal primary decompositions.
\end{theorem}

See \cite{Eisenbud94} for the next proposition. We will use it and the above theorem in the proof of Theorem \ref{decomposition}.

\begin{proposition} \label{Eisenbud}
Let $R$ be a Noetherian domain. If $f \in R$ and $f=u \prod p_i^{e_i}$, in such a way that u is a unit of $R$, the $p_i$ are primes generating distinct ideals $\langle p_i \rangle$, and each $e_i$ is a positive integer, then $\langle f \rangle =\cap \langle p_i^{e_i} \rangle$ is a minimal primary decomposition of $\langle f \rangle$.
\end{proposition}

The following theorem is from \cite{Liang14}. It looks like a special case of \mbox{Theorem \ref{decomposition}}, but is not. It is a stronger conclusion in zero-dimensional cases. We use it in proving \mbox{Theorem \ref{Lemma:dim}}, an equality property of the dimensions of ideals.

\begin{theorem}\label{DecForZero}
Let $K$ be an algebraically closed field, $>$ be a semigroup order in $A$ with $x_{j_1}<1,\ldots, x_{j_k} <1$ and $x_{j_{k+1}}>1,\ldots,x_{j_n} >1$ where $(j_1,\ldots,j_n)$ is a permutation of $(1,\ldots,n)$. Let $S=\{1+g:g=0 \vee \textsc{lt}(g)<1,g \in A\}$. Let $I \subset A$ be a zero-dimensional polynomial ideal and $I=\cap^k_{i=1}Q_i$ be its minimal primary decomposition. Let $P_i= \langle x_1-a_{i1},\ldots,x_n-a_{in} \rangle$ be the radical of $Q_i$ and suppose the $Q_i$ numbered so that $a_{ij_1}=a_{ij_2}=\cdots=a_{ij_k}=0$ for and only for the first $m$ $Q_i$.  Then, $S^{-1}I=\cap^m_{i=1}S^{-1}Q_i$ and $I^{ec}=(S^{-1}I)^c=\cap^m_{i=1}Q_i$ are minimal primary decompositions.
\end{theorem}

When proving Theorem \ref{Th:powerful}, we need the following concept which can be found in \cite{Kredel88}.

\begin{definition}[Strongly Independent Sets] \label{SIset}
Let $>$ be a global order on $\mathbb{Z}^n_{\geq 0}$ and $I$ be an ideal in $A$. A subset $u \subset x=\{x_1,\ldots,x_n\}$ is called a strongly independent set mod $I$ if $\langle \textsc{lt}(I) \rangle \cap K[u]=\{ 0 \}$.
\end{definition}

There is indeed a notion of independent set (cf. \cite{ Kredel88,Graebe93,Graebe95,Greuel07}), but we do not need it in this paper. The following proposition is an immediate corollary of Corollary 5.3.14 and Theorem 3.5.1(6) in \cite{Greuel07}. We will need it in proving Theorem \ref{Th:powerful}. The strongly independent sets can be computed by the \textsc{Singular} command \textsf{indepSet}.

\begin{proposition} \label{Prop:StronglyIndependentSet}
For an ideal $I \subset A$ and a global degree order $>$, the Krull dimension $\dim(A/I)$ is the maximal possible size of a strongly independent set mod $I$.
\end{proposition}

The proposition below can be found in \cite{Greuel07} as Theorem 3.5.1(1). It will be used in the proof of Theorem \ref{Lemma:dim}.

\begin{proposition} \label{Prop:A}
The Krull dimension of $A$ is $n$ and every maximal chain of prime ideals in $A$ have the same length $n$.
\end{proposition}

\section{Geometry of Control Orders} \label{Sec:Classify}

In this section, we first give the concept of control order, and then show the effect of a control order on controlling primary decompositions of polynomial ideals in \mbox{Theorem \ref{decomposition}}, where the geometry of a control order can be easily seen.

\begin{definition}[Control Orders]\label{Def:COrder}
Let $k$ be a non-negative integer with $k \leq n$. Let $>$ be a semigroup order such that $x_{j_1}<1,\ldots,x_{j_k}<1$ and $x_{j_{k+1}}>1,\ldots,x_{j_n}>1$ where $(j_1,\ldots,j_n)$ is a permutation of $(1,\ldots,n)$. If for any $t \in \textup{T}^{\{x_1,\ldots,x_n\}}$ we have that $x_j |t$ for some $j \in \{j_1,\ldots,j_k\}$ implies $t<1$, then we call $>$ a control order.
\end{definition}

By definition, local orders and global orders are control orders. If we take a local order $>_1$ on $\textup{T}^{\{x_{j_1},\ldots,x_{j_k}\}}$ and a global order $>_2$ on $\textup{T}^{\{x_{j_{k+1}},\ldots,x_{j_n}\}}$, then the block order $[>_1,>_2]$ is also a control order. But a control order is not necessary to be such a form.

\begin{example} \label{Example:COrder}
For a matrix $$M=\left( \begin{array}{ccc}
              -1 & -1 & 0 \\
               0 &  0 & 1 \\
               0 &  1 & 0
               \end{array} \right), $$
 there exists a semigroup order $>_M$ corresponding to it by the work of \cite{Robbiano1985}, since it is a total order, i.e., for any two tuples $u,v \in \mathbb{Z}^3_{\geq 0}$ we have $Mu^\textsc{t}=Mv^\textsc{t}$ implies $u=v$ where $Mu^\textsc{t}$ and $Mv^\textsc{t}$ are multiplications between matrices and column vectors. By the first two rows of $M$ we know that $>_M$ is a control order. Suppose the three variables are $x,y$ and $z$. Consider the two terms $x^2yz^2$ and $xy^2z$. We find that $x^2yz^2 >_M xy^2z$ which is determined by the second row of $M$, but $x^2y <_M xy^2$ by the third row. Therefore, $>_M$ is not a block order determined by a local order and a global order as shown above.
\end{example}

 A characterization of all the control orders is presented as Theorem \ref{Th:Matrix} at the end of this section. We can see in Corollary \ref{Corollary:COrder} that if only one variable is local under a control order, then this order is such a block order.

\begin{lemma} \label{Lemma:S}
For a semigroup order $>$ on $\textup{T}^{\{x_1,\ldots,x_n\}}$ with local variables $x_{j_1},\ldots,x_{j_k}$, we have that  $>$ is a control order if and only if $S:=\{1+g \in A:g=0 \ \textup{or} \ \textsc{lt}(g)<1\}=\{1+g:g \in \langle x_{j_1}, \dots, x_{j_k} \rangle\}$.
\end{lemma}

\begin{proof}
 For a polynomial $g \in A \setminus \{0\}$ and a semigroup order $>$, we have that $\textsc{lt}(g)<1$ $\Longrightarrow$ any term of $g$ is smaller than $1$ $\Longrightarrow$ any term of $g$ can be divided by a local variable $\Longrightarrow$ $g \in \langle x_{j_1}, \dots, x_{j_k} \rangle$. Thus, $S \subset \{1+g:g \in \langle x_{j_1}, \dots, x_{j_k} \rangle\}$. Then, $S = \{1+g:g \in \langle x_{j_1}, \dots, x_{j_k} \rangle\}$ if and only if $S \supset \{1+g:g \in \langle x_{j_1}, \dots, x_{j_k} \rangle\}$, if and only if for any $g \in \langle x_{j_1}, \dots, x_{j_k} \rangle \setminus \{0\}$ we have $\textsc{lt}(g)<1$, if and only if for any term $t \in \langle x_{j_1}, \dots, x_{j_k} \rangle$ we have $t<1$, if and only if $x_j | t$ for some $j \in \{x_{j_1},\ldots,x_{j_k}\}$ implies $t<1$, if and only if $>$ is a control order.
\end{proof}

The name of control orders is because of the following fact.

\begin{theorem}\label{decomposition}
Let $>$ be a semigroup order in $A$ with $x_{j_1}<1,\ldots, x_{j_k} <1$ and $x_{j_{k+1}}>1,\ldots,x_{j_n} >1$ where $(j_1,\ldots,j_n)$ is a permutation of $(1,\ldots,n)$. Let $S=\{1+g:g=0 \vee \textsc{lt}(g)<1,g \in A\}$. Let $I \subset A$ be a polynomial ideal and $I=\cap^k_{i=1}Q_i$ be its minimal primary decomposition. Suppose the $Q_i$ are numbered so that $Q_i + \langle x_{j_1}, \dots, x_{j_k} \rangle \neq A$ for and only for the first $m$ $Q_i$.  Then, $>$ is a control order if and only if for any ideal $I \subset A$ we have that $S^{-1}I=\cap^m_{i=1}S^{-1}Q_i$ and $I^{ec}=(S^{-1}I)^c=\cap^m_{i=1}Q_i$ are minimal primary decompositions.
\end{theorem}

\begin{pf}
``$\Rightarrow$" Suppose $>$ is a control order. Then $S=\{1+g \in A:g=0 \ \textup{or} \ \textsc{lt}(g)<1\}=\{1+g:g \in \langle x_{j_1}, \dots, x_{j_k} \rangle\}$ by Lemma \ref{Lemma:S}. For any ideal $I \subset A$, if $Q_i + \langle x_{j_1}, \ldots, x_{j_k} \rangle = A$ then we have $(1-r) \in Q_i$ where $r \in \langle x_{j_1}, \dots, x_{j_k} \rangle$. Hence, $Q_i \cap S \neq \emptyset$ since $(1-r)$ also belongs to $S$. Conversely, if $Q_i \cap S \neq \emptyset$, then we can take $g \in \langle x_{j_1},\ldots,x_{j_k} \rangle$ such that $1+g \in Q_i \cap S$. Thus, $Q_i + \langle x_{j_1}, \ldots, x_{j_k} \rangle = A$. Therefore, $Q_i + \langle x_{j_1}, \dots, x_{j_k} \rangle \neq A$ if and only if $Q_i \cap S = \emptyset $ if and only if $\sqrt{Q_i} \cap S = \emptyset$. By Theorem \ref{Atiyah}, we have that $S^{-1}I=\cap^m_{i=1}S^{-1}Q_i$ and $I^{ec}=(S^{-1}I)^c=\cap^m_{i=1}Q_i$ are minimal primary decompositions.

``$\Leftarrow$" Suppose $>$ is a semigroup order in $A$ with $x_{j_1}<1,\ldots, x_{j_k} <1$ and $x_{j_{k+1}}>1,\ldots,x_{j_n} >1$ where $(j_1,\ldots,j_n)$ is a permutation of $(1,\ldots,n)$ such that for every ideal $I \subset A$ we have $S^{-1}I=\cap^m_{i=1}S^{-1}Q_i$ and $I^{ec}=(S^{-1}I)^c=\cap^m_{i=1}Q_i$ are minimal primary decompositions. If $>$ is not a control order, then there exists a term $t:=t_1t_2$ with $t_1 \in \textup{T}^{\{x_{j_1},\ldots,x_{j_k} \} } \setminus \{1\}$ and $t_2 \in \textup{T}^{\{x_{j_{k+1}},\ldots,x_{j_n} \} } \setminus \{1\}$ such that $t_1t_2 >1$. Take $I=\langle t-1 \rangle$. Then $I + \langle x_{j_1}, \dots, x_{j_k} \rangle = A$, and consequently for any primary component $Q$ of $I$ we have that $Q+\langle x_{j_1}, \dots, x_{j_k} \rangle = A$ since $Q \supset I$. Hence,  $I^{ec}= A$ by the assumption. On the other hand, by Proposition \ref{Eisenbud}, we know that a set of factors of $t-1$ generate the primary ideals in one of its minimal primary decompositions. Among them, $t-1$ has a factor $p^\alpha$ with $\textsc{lt}(p^\alpha)>1$. Thus, $\langle p^{\alpha} \rangle \cap S = \emptyset$ since no elements in $\langle p^{\alpha} \rangle$ have leading term $1$, and consequently, $\langle p^{\alpha} \rangle ^{ec}=\langle p^{\alpha} \rangle$ is in a minimal decomposition of $I$ by \mbox{Theorem \ref{Atiyah}}, a contradiction. Therefore, $>$ is a control order.
\end{pf}

\begin{remark}
The geometry of Theorem \ref{decomposition} can be seen in $\bar{K}^n$. By the Nullstellensatz, we have that $Q_i + \langle x_{j_1}, \dots, x_{j_k} \rangle \neq A$ is equivalent to $\textbf{V}(Q_i) \cap \textbf{V}(\langle x_{j_1}, \dots, x_{j_k} \rangle) \neq \emptyset$. Thus, Theorem \ref{decomposition} means that only the components $\textbf{V}(Q_i)$ that intersect the linear variety $\textbf{V}(\langle x_{j_1}, \dots, x_{j_k} \rangle)$ are kept, and the others are discarded. This process is controlled by a control order. We can also say the process is controlled by $\langle x_{j_1}, \dots, x_{j_k} \rangle$. In the next section we will see that the latter is better and can be generalized.
\end{remark}

\begin{remark}
Though Theorem \ref{decomposition} seems to be more general than Theorem \ref{DecForZero} in the sense that it can deal with ideals with any dimension, in fact, if we restrict to zero-dimensional cases, the result is weaker, since it only uses control orders but Theorem \ref{DecForZero} considers all the semigroup orders.
\end{remark}

As what we did for semigroup orders in \cite{Liang14}, control orders can also be classified only according to the comparisons between variables and $1$. We give this property a name.

\begin{definition}[Characteristics]
Given $>$ a semigroup order on $\textup{T}^{\{x_1,\ldots, x_n\}}$ and $v$ an $n$-tuple with each entry $1$ or $-1$, we say $>$ has characteristic $v$ if for each $i=1,\ldots,n$ we have $x_i>1$ if and only if $v(i)=1$ (or $x_i<1$ if and only if $v(i)=-1$). For each $i$, we say $x_i$ has global (or local) characteristic under $>$ if $x_i>1$ (or $x_i<1$).
\end{definition}

Note that the characteristic of a variable under a semigroup order has only two values, i.e. global and local. With this definition, it is easier to express the following analogue of Corollary 3.3 in \cite{Liang14}.

\begin{corollary}\label{equivalence}
For two control orders $>_1$ and $>_2$ in $A$, they have the same effect on the localization of any ideal $I$, i.e., $(S_1^{-1}I)^c=(S_2^{-1}I)^c$, if and only if every variable has the same characteristic under $>_1$ and $>_2$. In this case, we say that $>_1$ and $>_2$ are equivalent.
\end{corollary}

\begin{pf}
``$\Leftarrow$" By Lemma \ref{Lemma:S}, we know that $S_1=S_2=\{1+g:g \in \langle x_{j_1}, \dots, x_{j_k} \rangle\}$ where $x_{j_i}$ ($i=1,\ldots,k$) are all variables with local characteristic.

``$\Rightarrow$" If there exists a variable $x_j$ with different characteristics under $>_1$ and $>_2$, say $x_j>_1 1$ and $x_j <_2 1$, then consider the ideal $I:=\langle x_j -1 \rangle$ in $A$. We can see that $(S_1^{-1}I)^c= I$ but $(S_2^{-1}I)^c=A$ by Theorem \ref{decomposition}, a contradiction.
\end{pf}

Corollary \ref{equivalence} says that control orders in the same equivalence class have the same effect in localizing rings. Thus, when using them, we only need to choose a representative of the control orders in an equivalence class. Especially, it is easy to construct such a representative by using the characteristics of variables under this control order. We can collect all the local variables and construct an arbitrary local order $>_1$ on the term set they generate, and construct a global order $>_2$ on the term set generated by the other variables. Then, the block order $[>_1,>_2]$ is a control order that we want.

\begin{theorem} \label{Th:Matrix}
Let $>$ be a mixed semigroup order on $\textup{T}^{\{x_1,\ldots,x_n\}}$ and $M$ be a $k \times n$ real matrix such that $>$ is equal to $>_M$. Every column of $M$ has a first nonzero entry from top to bottom. We call the row number of this entry the level of this column. When this entry is positive (or negative), we say that this column is global (or local). Then $>$ is a control order if and only if every local column has smaller level than every global column in $M$.
\end{theorem}

\begin{proof}
``$\Leftarrow$" For a term $t \in \textup{T}^{\{x_1,\ldots,x_n\}}$ corresponding to a tuple $u \in \mathbb{Z}_{\geq 0} ^n$, if $x_j |t$ for some local variable $x_j$, then the set $X_{local}:=\{x_i : x_i |t, \textup{$x_i$ is a local variable} \}$ is not empty. Take a maximal subset $X^* \subset X_{local}$ such that among the variables in $X_{local}$, each variable in $X^*$ corresponds to a local column with the minimal level. Suppose the row of $M$ corresponding to this level is $w$, then $wu^\textsc{t}<0$ which means $t<1$. Therefore, $>$ is a control order.

``$\Rightarrow$" Suppose $>$ is a control order and there exists a local column with no smaller level than a global column of $M$, and they have the respective first nonzero entries $-b_1$ and $b_2$ where $b_1$ and $b_2$ are positive real numbers. Let $x$ and $y$ be the two variables corresponding to the local and global columns, respectively. Then, we find that $xy^{\lceil b_1/b_2 \rceil+1} >1$, a contradiction.
\end{proof}

\begin{corollary} \label{Corollary:COrder}
Let $>$ be a mixed semigroup order on $\textup{T}^{\{x_1,\ldots,x_n\}}$ with local variables $x_{j_1},\ldots,x_{j_k}$, and $>_1$ and $>_2$ are local and global orders as restrictions of $>$ on $\textup{T}^{\{x_{j_1},\ldots,x_{j_k}\}}$ and $\textup{T}^{\{x_{1},\ldots,x_n\} \setminus \{x_{j_1},\ldots,x_{j_k}\}}$, respectively. Then, $>$ is a control order implies it is the block order $[>_1,>_2]$ if and only if $k=1$.
\end{corollary}

\begin{proof}
``$\Leftarrow$" Suppose $>$ is a control order and $k=1$. By Theorem \ref{Th:Matrix}, there is a real matrix $M$ such that $>$ equals to $>_M$ and the first nonzero row $w$ of $M$ contains only one nonzero entry $-b$ where $b>0$. Since the upper rows of $M$ can be used to reduce the lower rows and this process does not change the order $>_M$, we can use $w$ to reduce all the rows below it in $M$ so that the column that this entry locates in has only one nonzero entry. Thus, we obtain a new matrix $M^*$ that represents the block order $[>_1,>_2]$.

``$\Rightarrow$" If $k \geq 2$, then we can use Example \ref{Example:COrder} to construct a control order that is not the block order described above. Let $x_1,\ldots,x_k$ be local variables and $x_{k+1},\ldots,x_n$ be global variables. Take a local order $>_{lcoal}$ on $\textup{T}^{\{x_1,\ldots,x_{k-2}\}}$ and a global order $>_{global}$ on $\textup{T}^{\{x_{k+2},\ldots,x_n\}}$ (they can be empty). Denote the control order described in Example \ref{Example:COrder} by $>_{mixed}$ and assume that it is on $\textup{T}^{\{x_{k-1},x_k,x_{k+1}\}}$. Then, the block order $[>_{lcoal},>_{mixed},>_{global}]$ is a control order by Theorem \ref{Th:Matrix}, but it is not a block order we want in this corollary for the same reason as in Example \ref{Example:COrder}.
\end{proof}

\section{Extractions of Ideals} \label{Sec:ControlVarieties}

In the last section, we see that the primary decomposition is controlled by a control order or a special ideal. In this section, we study how to control a primary decomposition by using arbitrary ideals.

For two polynomial ideals $I=\langle f_1,\ldots,f_{r} \rangle $ and $J=\langle g_1,\ldots,g_{s} \rangle$ in $A$, we want to keep the components of $\mathbf{V}(I) \subset \bar{K}^n$ that intersect $\mathbf{V}(J) \subset \bar{K}^n$ and delete the ones that do not intersect $\mathbf{V}(J)$. In other words, we use the ideal $J$ or more precisely the variety $\mathbf{V}(J)$ to control the process. In this case, we call $J$ and $\mathbf{V}(J)$ the control ideal and control variety of $I$, respectively. As what we did in \cite{Liang14}, we need to rename $g_1,\ldots,g_{s}$ as new variables $t_1,\ldots,t_{s}$ by introducing new relations $t_1-g_1,\ldots,t_{s}-g_{s}$ into $I$ to obtain a larger ideal $I':=\langle I, t_1-g_1,\ldots,t_{s}-g_{s} \rangle $ in a larger ring $A':=K[x_1,\ldots,x_n,t_1,\ldots,t_{s}]$.

Now, we study the relation of $I$ and $I'$ below.

\begin{theorem}\label{Th:lift}
If $I'=\cap_{i=1}^k{\tilde{Q}_i}$ is a minimal primary decomposition of $I'$ in $A'$, then $I=\cap_{i=1}^k{\tilde{Q}_i}|_{t_1=g_1,\ldots,t_{s}=g_{s}}$ holds and is a minimal primary decomposition of $I$ in $A$.
\end{theorem}

\begin{proof}
We first prove $\tilde{Q}_i \cap A = \tilde{Q}_i | _{t_1=g_1,\ldots,t_{s}=g_{s}}$. It is easy to see $\tilde{Q}_i \cap A \subset \tilde{Q}_i | _{t_1=g_1,\ldots,t_{s}=g_{s}}$ since for every $f \in \tilde{Q}_i \cap A$ we have $f=f| _{t_1=g_1,\ldots,t_{s}=g_{s}} \in \tilde{Q}_i | _{t_1=g_1,\ldots,t_{s}=g_{s}}$ ($f$ contains no $t_i$). Conversely, for an $f \in \tilde{Q}_i | _{t_1=g_1,\ldots,t_{s}=g_{s}} \subset A$ there exists a polynomial $\tilde{f} \in \tilde{Q}_i$ such that $f=\tilde{f}|_{t_1=g_1,\ldots,t_{s}=g_{s}}$. Note that $\tilde{f}|_{t_1=g_1,\ldots,t_{s}=g_{s}}$ is the remainder of $\tilde{f}$ divided by $\{t_1-g_1,\ldots,t_{s}-g_{s}\}$ subsequently, and $\{t_1-g_1,\ldots,t_{s}-g_{s}\} \subset I' \subset \tilde{Q}_i$. Thus, $f=\tilde{f}|_{t_1=g_1,\ldots,t_{s}=g_{s}} \in \tilde{Q}_i$, i.e., $f \in \tilde{Q}_i \cap A$. Thus, $\tilde{Q}_i \cap A = \tilde{Q}_i | _{t_1=g_1,\ldots,t_{s}=g_{s}}$.

Now, we prove the equality in the conclusion of the theorem. We have that $I=I' \cap A = (\cap_{i=1}^k{\tilde{Q}_i}) \cap A = \cap_{i=1}^k{(\tilde{Q}_i \cap A)} = \cap_{i=1}^k{\tilde{Q}_i}|_{t_1=g_1,\ldots,t_{s}=g_{s}}$.

Next, we prove ${\tilde{Q}_i}|_{t_1=g_1,\ldots,t_{s}=g_{s}}$ is primary. To see it is an ideal in $A$ is trivial. Suppose that $p$ and $q$ are two polynomials in $A$ and $pq \in {\tilde{Q}_i}|_{t_1=g_1,\ldots,t_{s}=g_{s}} \subset \tilde{Q}_i$. If $p$ is not in ${\tilde{Q}_i}|_{t_1=g_1,\ldots,t_{s}=g_{s}}$, then it is not in $\tilde{Q}_i$ either. Since $\tilde{Q}_i$ is primary in $A'$, there exists a positive integer $k$ such that $q^k \in \tilde{Q}_i$. Then $q^k \in \tilde{Q}_i \cap A = {\tilde{Q}_i}|_{t_1=g_1,\ldots,t_{s}=g_{s}}$, i.e., ${\tilde{Q}_i}|_{t_1=g_1,\ldots,t_{s}=g_{s}}$ is primary in $A$.

Finally, it is only needed to show the primary decomposition $I=\cap_{i=1}^k{\tilde{Q}_i}|_{t_1=g_1,\ldots,t_{s}=g_{s}}$ is minimal. i) The radicals of ${\tilde{Q}_i}|_{t_1=g_1,\ldots,t_{s}=g_{s}}$ are distinct. Otherwise, there exists $i$ and $j$ such that $i \neq j$ and the radicals $P_i$ and $P_j$ of ${\tilde{Q}_i}|_{t_1=g_1,\ldots,t_{s}=g_{s}}$ and ${\tilde{Q}_j}|_{t_1=g_1,\ldots,t_{s}=g_{s}}$ are equal. Denote the radicals of $\tilde{Q}_i$ and $\tilde{Q}_j$ by $\tilde{P}_i$ and $\tilde{P}_j$, respectively. For every $\tilde{f} \in \tilde{Q}_i$, it can be written as
\begin{equation}
\tilde{f}=\tilde{f}|_{t_1=g_1,\ldots,t_{s}=g_{s}} + \sum_{w=1}^{s} r_w(t_w-g_w)
\end{equation}
where $r_w \in A'$. Then $\tilde{f} \in \tilde{P}_j$, since $\tilde{f}|_{t_1=g_1,\ldots,t_{s}=g_{s}} \in {\tilde{Q}_i}|_{t_1=g_1,\ldots,t_{s}=g_{s}} \subset P_i = P_j \subset \tilde{P}_j$ and $\{t_1-g_1,\ldots,t_{s}-g_{s}\} \subset I' \subset \tilde{Q}_j$. Thus, we have $\tilde{Q}_i \subset \tilde{P}_j$. So, $\tilde{P}_i \subset \tilde{P}_j$. Similarly, we can obtain $\tilde{P}_i \supset \tilde{P}_j$. Thus, $\tilde{P}_i = \tilde{P}_j$, a contradiction. Therefore, all radicals of ${\tilde{Q}_i}|_{t_1=g_1,\ldots,t_{s}=g_{s}}$ are distinct. ii) If there exists an $i$ such that ${\tilde{Q}_i}|_{t_1=g_1,\ldots,t_{s}=g_{s}} \supset \cap_{j \neq i}{\tilde{Q}_j}|_{t_1=g_1,\ldots,t_{s}=g_{s}}$, then by formula (1) we know that ${\tilde{Q}_i} \supset \cap_{j \neq i}{\tilde{Q}_j}$, a contradiction. Thus, for every $i$ we have that ${\tilde{Q}_i}|_{t_1=g_1,\ldots,t_{s}=g_{s}} \not \supset \cap_{j \neq i}{\tilde{Q}_j}|_{t_1=g_1,\ldots,t_{s}=g_{s}}$. By i) and ii), we conclude that the primary decomposition $I=\cap_{i=1}^k{\tilde{Q}_i}|_{t_1=g_1,\ldots,t_{s}=g_{s}}$ is minimal.
\end{proof}

\begin{definition}[Extractions]\label{Def:CPart}
Let $I=\langle f_1,\ldots,f_r \rangle$ and $J=\langle g_1,\ldots,g_s \rangle $ be two polynomial ideals in $A$. Let $I=\cap_{i=1}^k Q_i$ be a minimal primary decomposition of $I$ in $A$ and $Q_i$ be numbered so that $Q_i + J \neq A$ for and only for the first $m$ $Q_i$. We call $\beta(I,J):=\cap_{i=1}^m Q_i$ the extraction of $I$ by $J$.
\end{definition}

\begin{remark}
By the Nullstellensatz, $Q_i +J \neq A$ is equivalent to $\mathbf{V}(Q_i) \cap \mathbf{V}(J) \neq \emptyset$ in $\bar{K}^n$. It means that in this case, $\beta(I,J)$ is just the intersection of the primary components whose varieties meet $\mathbf{V}(J)$. To emphasize $\bar{K}$, we denote $\mathbf{V}(\cdot)$ in $\bar{K}^n$ by $\mathbf{V}_{\bar{K}}(\cdot)$ in the rest of this paper.
\end{remark}

\begin{remark}
When $m=0$ in Definition \ref{Def:CPart}, $\beta(I,J)=\cap_{Q \in \emptyset} Q = A$ by knowledge of set theory (ZFC, Zermelo-Fraenkel set theory with the axiom of choice).
\end{remark}

\begin{proposition} \label{Prop:WellDefined}
The ideal $\beta(I,J)$ is well defined.
\end{proposition}

\begin{proof}
If $Q_i + J \neq A$ then $\sqrt{Q_i} + J \neq A$.  For a $Q_j$ with $\sqrt{Q_j} \subset \sqrt{Q_i}$, it is easy to see $\sqrt{Q_j} + J \neq A$. Since $Q_j + J \neq A$ if and only if $\sqrt{Q_j} + J \neq A$, we have $Q_j + J \neq A$. This means that for every minimal primary decomposition of $I$, all the primes belonging to $\beta(I,J)$ conform an isolated set of primes belonging to $I$. Then, by the second uniqueness theorem on page 54 in \cite{Atiyah1969}, $\beta(I,J)$ is independent of the decomposition.
\end{proof}

\begin{corollary}
For any ideal $I \subset A$ , the set $\{\beta(I,J): J \textup{ is an ideal in } A\}$ is finite.
\end{corollary}

\begin{proof}
Note that $I$ has only finitely many isolated sets of prime ideals.
\end{proof}

\begin{remark}
If we define $\alpha(I,J):=\cap_{i=m+1}^k Q_i$ in Definition \ref{Def:CPart}, then we will see that it is not well defined. For example, take $I=\langle x^2 -xy \rangle$ and $J=\langle x , y-1\rangle$ in $K[x,y]$. Let $I=\langle x \rangle \cap \langle x, y \rangle^2 = \langle x \rangle \cap \langle x^2 ,y \rangle $ be two minimal primary decompositions in $A$. We can see in the two cases the values of $\beta (I,J)$ are all equal to $\langle x \rangle$. However, $\alpha_1(I,J) = \langle x, y \rangle^2 \neq \langle x^2 ,y \rangle = \alpha_2(I,J)$.
\end{remark}

As a consequence of Theorem \ref{Atiyah}, we can see that for any multiplicatively closed subset $S \subset A$ and any polynomial ideal $I$, the associated primes of $(S^{-1}I)^c$ constitute an isolated set of primes of $I$. So, the following theorem means that extractions of a polynomial ideal are as powerful as localizations.

\begin{theorem} \label{Th:powerful}
Let $I= \cap_{i=1}^k Q_i$ be a minimal primary decomposition of an ideal $I$ in $A$ and $\{ \sqrt{Q_{i_1}}, \ldots, \sqrt{Q_{i_m}} \}$ be an isolated set of prime ideals of $I$. Then there exists an ideal $J \in A$ such that $\beta(I,J)= Q_{i_1} \cap \cdots \cap Q_{i_m}$. Namely, $\{ \beta (I,J): J \textup{ is an ideal in $A$} \} = \{Q_{i_1} \cap \cdots \cap Q_{i_m}: \{ \sqrt{Q_{i_1}}, \ldots, \sqrt{Q_{i_m}} \} \textup{ is an isolated set of primes of $I$}  \}$.
\end{theorem}

\begin{proof}
Let $\{P_1,\ldots,P_j\} $ be maximal elements of $\{ \sqrt{Q_{i_1}}, \ldots, \sqrt{Q_{i_m}} \}$. Then $P_1, \ldots, P_j$ are prime ideals. For every $t \in \{1,\ldots,j\}$, $P_t \not \supset \cap_{i \in \{1,\ldots,k\} \setminus \{i_1,\ldots,i_m\}} \sqrt{Q_i}$ since otherwise $P_t \supset \sqrt{Q_i}$ for some $i \in \{1,\ldots,k\} \setminus \{i_1,\ldots,i_m\}$, a contradiction. Then by the Nullstellensatz $\mathbf{V}_{\bar{K}}(P_t) \not \subset \cup_{i \in \{1,\ldots,k\} \setminus \{i_1,\ldots,i_m\}} \mathbf{V}_{\bar{K}}(Q_i)$ where $\bar{K}$ is the algebraic closure of $K$. Take a point $p \in \mathbf{V}_{\bar{K}}(P_t) \setminus \cup_{i \in \{1,\ldots,k\} \setminus \{i_1,\ldots,i_m\}} \mathbf{V}_{\bar{K}}(Q_i)$. Let $M_t$ be the maximal subset of $A$ such that all its elements vanish at $p$. Then $M_t$ is a proper and nonempty ideal of $A$. We prove $M_t$ is a maximal ideal. Firstly, $M_t$ is zero-dimensional. Otherwise, by \mbox{Proposition \ref{Prop:StronglyIndependentSet}}, we can get a nonempty strongly independent set of $M_t$ w.r.t. an arbitrary global degree order on $\textup{T}^{\{x_1,\ldots,x_n\}}$. For simplicity, suppose this set is $\{x_{1},\ldots,x_{w}\}$. Denote the minimal polynomial of $x_{h}(p)$ in $K[x_h]$ by $f_h$. Then we obtain a larger set $M_t \cup \{f_h:h=1,\ldots,w\}$ whose elements vanish at $p$, a contradiction. Next, we prove that $M_t$ is a prime ideal. Otherwise, we have a minimal primary decomposition of $M_t$ in $A$. Take an associated prime that vanishes at $p$ as $P$. Then $P$ strictly contains $M_t$, a contradiction. Thus, $M_t$ is a zero-dimensional prime ideal, i.e. a maximal ideal. Then, $M_t$ is the maximal set in $A$ whose elements vanish at any fixed point in $\mathbf{V}_{\bar{K}}(M_t)$. Take $J=\cap_{t=1}^j M_t$. It is easy to check that $Q_u +J \neq A$ for every $u \in \{i_1,\ldots,i_m\}$.  For $v \in \{1,\ldots,k\} \setminus \{i_1,\ldots,i_m\}$ and every $t$, we have $Q_v+M_t=A$, otherwise $\mathbf{V}_{\bar{K}}(Q_v) \cap \mathbf{V}_{\bar{K}}(M_t) \neq \emptyset$ and thus $Q_v \subset M_t$ which implies $\mathbf{V}_{\bar{K}}(Q_v) \supset \mathbf{V}_{\bar{K}}(M_t)$, a contradiction. So, $Q_v +J =A$. Thus, $\beta(I,J)= Q_{i_1} \cap \cdots \cap Q_{i_m}$. Check the proof of Proposition \ref{Prop:WellDefined} to see why the two sets in the last sentence of this theorem are equal.
\end{proof}

The next theorem provides a relation between $\beta(I,J)$ and $I'$.

\begin{theorem} \label{Th:Projection}
Let $I=\langle f_1,\ldots,f_r \rangle$ and $J=\langle g_1,\ldots,g_s \rangle $ be two polynomial ideals in $A$. Let $>$ be a control order on $\textup{T}^{\{x_1,\ldots,x_{n+s}\}}$ with local variables $x_{n+1},\ldots,x_{n+s}$ and global variables $x_{1},\ldots,x_{n}$. Let $I'=\langle I, x_{n+1}-g_1,\ldots,x_{n+s}-g_s \rangle \subset A'$. Suppose $I'=\cap_{i=1}^k \tilde{Q}_i$ is a minimal primary decomposition of $I'$ in $A'$ and $\tilde{Q}_i$ are numbered so that $\tilde{Q}_i + \langle x_{n+1}, \ldots, x_{n+s} \rangle \neq A'$ for and only for the first $m$ $\tilde{Q}_i$. Then, we have that $\beta(I,J)=I'^{ec}|_{x_{n+1}=g_1,\ldots, x_{n+s}=g_s} = \cap_{i=1}^m \tilde{Q}_i|_{x_{n+1}=g_1,\ldots, x_{n+s}=g_s}$ is a minimal primary decomposition.
\end{theorem}

\begin{proof}
By Theorem \ref{decomposition}, we know that $I'^{ec}=\cap_{i=1}^m \tilde{Q}_i$ is a minimal primary decomposition of $I'^{ec}$. Then, by Theorem \ref{Th:lift}, $I'^{ec}|_{x_{n+1}=g_1,\ldots, x_{n+s}=g_s}=\cap_{i=1}^m \tilde{Q}_i|_{x_{n+1}=g_1,\ldots, x_{n+s}=g_s}$ is a minimal primary decomposition. So, we only need to prove the equality that $\beta(I,J)=I'^{ec}|_{x_{n+1}=g_1,\ldots, x_{n+s}=g_s}$. By Theorem \ref{Th:lift}, $I=\cap_{i=1}^k \tilde{Q}_i|_{x_{n+1}=g_1,\ldots, x_{n+s}=g_s}$ is a minimal primary decomposition. For simplicity, denote $\tilde{Q}_i|_{x_{n+1}=g_1,\ldots, x_{n+s}=g_s}$ by $Q_i$.  By \mbox{Definition \ref{Def:CPart}}, we need to show $Q_i + J =A$ if and only if $\tilde{Q}_i +\langle x_{n+1}, \ldots, x_{n+s} \rangle = A'$ for every $i=1,\ldots,k$. If $Q_i + J =A$, then there exists a polynomial $f \in Q_i$ such that $1=f+\sum_{t=1}^s u_tg_t$ where $u_t \in A$. Since $\{ x_{n+1}-g_i, \ldots, x_{n+s}-g_s \} \subset \tilde{Q}_i$ and $f \in Q_i \subset \tilde{Q}_i$ (see the proof of Theorem \ref{Th:lift}), we have $1=f+\sum_{t=1}^s u_t(g_t-x_{n+t}+x_{n+t})=f+\sum_{t=1}^s u_t(g_t-x_{n+t})+\sum_{t=1}^s u_t x_{n+t} \in \tilde{Q}_i + \langle x_{n+1}, \ldots, x_{n+s} \rangle$, i.e. $\tilde{Q}_i + \langle x_{n+1}, \ldots, x_{n+s} \rangle = A'$. Conversely, if $\tilde{Q}_i + \langle x_{n+1}, \ldots, x_{n+s} \rangle = A'$, then there exists a polynomial $\tilde{f} \in \tilde{Q}_i$ such that $1=\tilde{f}+\sum_{t=1}^s v_t x_{n+t}$ where $v_t \in A'$. Substituting $x_{n+1}=g_1, \ldots, x_{n+s}=g_s$ in this expression of $1$, we obtain $1=\tilde{f}|_{x_{n+1}=g_1,\ldots, x_{n+s}=g_s}+\sum_{t=1}^s v_t|_{x_{n+1}=g_1,\ldots, x_{n+s}=g_s} g_t \in Q_i + J$, i.e. $Q_i +J =A$. Therefore, the theorem has been proved.
\end{proof}

\begin{corollary} \label{Corollary:membership}
With the conditions in Theorem \ref{Th:Projection}, we have that $\beta(I,J)=I'^e \cap A$ and $\sqrt{\beta(I,J)}=\sqrt{I'^e} \cap A$.
\end{corollary}
\begin{proof}
By Theorem \ref{Th:Projection}, $\beta(I,J)= I'^{ec}|_{x_{n+1}=g_1,\ldots, x_{n+s}=g_s} = I'^{ec} \cap A = I'^e \cap A' \cap A = I'^e \cap A$. Now, we prove the other equality. Firstly, we have $\sqrt{\beta(I,J)} \subset \sqrt{I'^e}$ and $\sqrt{\beta(I,J)} \subset A$, and thus $\sqrt{\beta(I,J)} \subset \sqrt{I'^e} \cap A$. Secondly, for every $f \in A$, if $f \in \sqrt{I'^e}$ then there exists a positive integer $d$ such that $f^d \in I'^e \cap A = \beta(I,J)$. Hence, $\sqrt{I'^e} \cap A \subset \sqrt{\beta(I,J)}$. We are done.
\end{proof}

\begin{remark} \label{Remark:membership}
The above corollary can be used to determine the membership problems of $\beta(I,J)$ and $\sqrt{\beta(I,J)}$. This is because that for a polynomial $f \in A$, we have $f \in \beta(I,J)$ if only if $f \in I'^e$ and $f \in \sqrt{\beta(I,J)}$ if only if $f \in \sqrt{I'^e}$, and then we can determine the membership problems of $I'^e$ and $\sqrt{I'^e}$ by the standard basis method (cf. \cite{Alonso91,Graebe95,Cox05}).
\end{remark}

\section{Dimensions} \label{Sec:Dimension}
In this section, we study the Krull dimensions of an ideal in a localized ring $\textup{Loc}_>(A)$ and its contraction in $A$ as well as their relations to control orders.

To distinguish dimensions of ideals in different rings, instead of $\dim(W)$, we denote the dimension of an ideal $W$ in a ring $R$ by $\dim(R/W)$ as the Krull dimension of the quotient ring $R/W$. This is also the definition of the Krull dimension of an ideal.

The following example shows that $\dim(\textup{Loc}_>(A)/I^e)$ and $\dim(A/I^{ec})$ may not coincide for a general semigroup order $>$ and an arbitrary ideal $I$ in $A$.

\begin{example}
Consider a localized ideal $\langle xy-1 \rangle ^e \subset \textup{Loc}_>(K[x,y])$ where $>$ is a semigroup order given by a $2 \times 2$ diagonal  matrix $M$ with $M_{11}=1$ and $M_{22}=-1$. Then for every element $f$ in $\langle xy-1 \rangle$, we have $(xy)|\textsc{lt}(f)$ and thus $\langle f \rangle \cap S = \emptyset$. This implies that $\langle xy-1 \rangle ^{ec} =\langle xy-1 \rangle $ (note that $\langle xy-1 \rangle$ is a prime ideal) and has dimension $1$ in $K[x,y]$ by Proposition \ref{Prop:StronglyIndependentSet}. But $\langle xy-1 \rangle ^e$ is zero-dimensional, since it is a maximal ideal in $\textup{Loc}_>(K[x,y])$. To see this clearly, we verify that all the nonzero elements in $B:=\textup{Loc}_>(K[x,y]) / \langle xy-1 \rangle ^e=S^{-1} (K[x,y] / \langle xy-1 \rangle)$ are invertible. Note that every polynomial $g$ in $K[x,y]$ can be reduced to $g_1(x)+g_2(y)$ by $xy-1$ where $g_1, g_2$ are univariate polynomials with $g_2(0)=0$. When $g \not \in \langle xy-1 \rangle$, we have $g_1(x)+g_2(y) \neq 0$ in $A$. If $g_1$ is a zero polynomial, then $g_2$ is not zero in $A$ and can be factored as $cy^r(1+h(y))$ where $c$ is a nonzero constant, $r$ is a positive integer and $h$ is a univariate polynomial. Thus, $g_2$ is invertible ($y$ is a unit in $B$) and so does $g$ in this case. If $g_1$ is not a zero polynomial, then consider $y^{\deg(g_1)}(g_1+g_2)$. It can be reduced to a univariate polynomial in $y$ with a nonzero constant term $\textsc{lc}(g_1)$ (the restriction of $>$ on $\textup{T}^{\{x\}}$ is a global order) by $xy-1$ and hence invertible in $B$. Therefore, $B$ is a field and $\langle xy-1 \rangle ^e$ is a maximal ideal in $\textup{Loc}_>(K[x,y])$.
\end{example}

\begin{lemma} \label{Lemma:bar}
For an ideal $I \subset A$, we denote the ideal generated by $I$ in $\bar{K}[x_1,\ldots,x_n]$ by $\bar{I}$. Then, for any semigroup order $>$ on $\textup{T}^{\{x_1,\ldots,x_n\}}$, we have that $\bar{I}^{ec} \cap A = I^{ec}$.
\end{lemma}
\begin{proof}
For any semigroup order $>$, the computation of standard bases of $S^{-1}I$ and $\bar{S}^{-1}\bar{I}$ w.r.t. $>$ in the respective localized rings are the same, (note that $S$ and $\bar{S}$ are different sets). Then we can obtain a finite set of polynomials $F$ in $A$ as their standard bases. For a polynomial $f \in A$, we have that $f \in \bar{I}^e$ if and only if the weak normal form is $0$, if and only if $f \in I^e$. So, $\bar{I}^{e} \cap A = I^{e} \cap A$, i.e. $\bar{I}^{ec} \cap A = I^{ec}$.
\end{proof}

\begin{theorem} \label{Lemma:dim}
For a semigroup order $>$ on $\textup{T}^{\{x_1,\ldots,x_n\}}$, we have that $\dim(\textup{Loc}_>(A)/I^e) = \dim(A/I^{ec})$ for any ideal $I \subset A$ if and only if $>$ is a control order.
\end{theorem}

\begin{proof}
``$\Leftarrow$" Since $(I^{ec})^{e} = I^e$, we have that $\dim(\textup{Loc}_>(A)/I^e) \leq \dim(A/I^{ec})$. Now we prove the converse part. By definition, the Krull dimension of the ideal $I^{ec}$ is the Krull dimension of the quotient ring $A/I^{ec}$, i.e. the maximal length $l$ of a chain of primes containing $I^{ec}$. So, there exists a maximal chain of prime ideals containing $I^{ec}$ with length $l$, i.e., $I^{ec} \subset P_0 \varsubsetneq \cdots \varsubsetneq P_l$ where $P_0,\ldots,P_l$ are prime ideals in $A$. Suppose $I=\cap_{i=1}^k Q_i$ is a minimal primary decomposition and the $Q_i$ are numbered so that $Q_i + \langle x_{j_1}, \dots, x_{j_w} \rangle \neq A$ for and only for the first $m$ $Q_i$, where $x_{j_1},\ldots,x_{j_w}$ are all of the local variables. Then, by \mbox{Theorem \ref{decomposition}}, $I^{ec}=\cap_{i=1}^m Q_i$ is a minimal primary decomposition. Thus, all the minimal prime ideals containing $I^{ec}$ are among the primes $\sqrt{Q_1}, \ldots, \sqrt{Q_m}$, and there exists a $j \in \{1,\ldots,m\}$ such that $P_0=\sqrt{Q_j}$. Thus, we have that $P_0 + \langle x_{j_1}, \dots, x_{j_w} \rangle \neq A$ and $\textbf{V}_{\bar{K}}(P_0) \cap \textbf{V}_{\bar{K}}(\langle x_{j_1}, \dots, x_{j_w} \rangle)$ is not empty in $\bar{K}^n$. Take a point $p \in \textbf{V}_{\bar{K}}(P_0) \cap \textbf{V}_{\bar{K}}(\langle x_{j_1}, \dots, x_{j_w} \rangle)$. Then we can construct a maximal ideal $M_p$ in $A$ as the maximal subset of polynomials passing $p$ as in the proof of \mbox{Theorem \ref{Th:powerful}}. It is easy to see $M_p+\langle x_{j_1}, \ldots, x_{j_w} \rangle \neq A$ by the Nullstellensatz. We can obtain a maximal chain of primes containing $I^{ec}$ in the form $I^{ec} \subset W_0 \varsubsetneq \cdots \varsubsetneq W_{l^*}$ where $W_0,\ldots,W_{l^*}$ are all prime ideals with $W_0=P_0$ and $W_{l^*}=M_p$. Then Proposition \ref{Prop:A} implies $l=l^*$. Moreover, for every $i \in \{0,\ldots,l\}$ we know $W_i + \langle x_{j_1}, \dots, x_{j_w} \rangle \neq A$ which implies $I^e \subset {W_0}^e \varsubsetneq \cdots \varsubsetneq {W_{l}}^e$ is a chain of prime ideals containing $I^e$ in $\textup{Loc}_>(A)$. So, $\dim(\textup{Loc}_>(A)/I^e) \geq \dim(A/I^{ec})$.

``$\Rightarrow$" Suppose $>$ is not a control order and $x_{j_1},\ldots,x_{j_w}$ are all of the local variables. Then there exists a term $t \in \textup{T}^{\{x_1,\ldots,x_n\}}$ such that it is larger than $1$ and can be divided by a local variable $x_j$. Take an irreducible (also prime) factor $f$ of $t-1$ such that $\textsc{lt}(f)>1$. We know that $\langle f \rangle$ is prime, $\langle f \rangle ^{ec} = \langle f \rangle $ and $\textbf{V}_{\bar{K}}(f) \cap \textbf{V}_{\bar{K}}(\langle x_{j_1},\ldots,x_{j_w} \rangle) = \emptyset$ (otherwise, $t-1$ can vanish at a point when $x_{j_1}= \cdots = x_{j_w} = 0$, a contradiction with $x_j | t$). Every maximal chain $C_1$ of prime ideals containing $\langle f \rangle ^e$ corresponds to a chain $C_2$ of primes containing $\langle f \rangle$ in $A$. However, $C_2$ is not maximal in $A$, since the maximal prime ideal in $C_2$ is not a maximal ideal in $A$. (Otherwise, consider a maximal ideal $M$ in $C_2$. By definition, $M$ is a zero-dimensional ideal. Since $\bar{M}$ and $M$ have the same generating set, we know that $\textbf{V}_{\bar{K}}(\bar{M}) = \textbf{V}_{\bar{K}}(M) \subset \textbf{V}_{\bar{K}}(f) \subset \textbf{V}_{\bar{K}}(t-1)$ and consequently $\textbf{V}_{\bar{K}}(\bar{M})$ consists of finitely many isolated points in $\bar{K}^n$.  Then we have a minimal primary decomposition $\bar{M} = \cap_{i=1}^m Q_i$ in $\bar{K}[x_1,\ldots,x_n]$ where $Q_i$ are primary ideals corresponding to each point in $\textbf{V}_{\bar{K}}(\bar{M})$. By Theorem \ref{DecForZero}, $\bar{M}^{ec} = \bar{K}[x_1,\ldots,x_n]$. Thus, $M^{ec} = \bar{M}^{ec} \cap A = A$ by \mbox{Lemma \ref{Lemma:bar}}, a contradiction.) Thus, $C_2$ can be extended to another chain $C_3$ of primes containing $\langle f \rangle$ with $\# C_3 \geq \# C_2 +1$. So, $\dim(A/\langle f \rangle ^{ec}) > \dim(\textup{Loc}_>(A)/\langle f \rangle ^e)$, a contradiction.
\end{proof}

\begin{corollary} \label{Corollary:dim}
Let $>$ be a control order as in Theorem \ref{Th:Projection}. Then, the Krull dimensions of $\beta(I,J)$ and $I'^e$ are equal.
\end{corollary}

\begin{proof}
It is obviously true by Theorem \ref{Th:Projection} and Theorem \ref{Lemma:dim}.
\end{proof}

\section{Basic Properties of Extractions} \label{Sec:Basic}

In this section, we prove some equalities about different combinations of extractions and other basic ideal operations on polynomial ideals. Before introducing these results, we need a lemma which is used in proving the third equality in Proposition \ref{Prop:BasicProperties}.

\begin{lemma}\label{Lemma:E_CPart}
Let $I$ and $J$ be two polynomial ideals in $A$, and $I=\cap_{i=1}^k Q_i$ a primary decomposition (not necessarily minimal). Let $E=\{Q_i: Q_i+J \neq A\}$. Then $\beta(I,J)=\cap_{Q \in E} Q$.
\end{lemma}

\begin{proof}
We can arrange this decomposition to get a minimal one. The first step is to intersect all the primary ideals $Q_i$ with the same radicals to obtain a primary decomposition $I=\cap_{i=1}^l M_i$ such that $\sqrt{M_i}$ are distinct. Next, delete the redundant primary ideals $M_j$ with $M_j \supset \cap_{i \neq j}M_i$ to get a minimal primary decomposition $I=\cap_{i=1}^r M_i$ (we suppose $M_{r+1} , \ldots, M_{l}$ are redundant ones).

In the first step, suppose $M_i = Q_{i_1} \cap \cdots \cap Q_{i_k}$ and let $F:=\{M_i: M_i +J \neq A, i\in \{1,\ldots,l\}\}$. We can easily see that $M_i+J \neq A$ if and only if $Q_t + J \neq A$ for an arbitrary $t \in \{i_1,\ldots,i_k\}$. Thus, $\cap_{Q \in E} Q=\cap_{M \in F} M$.

In the second step, we have $\cap_{i=1}^r M_i=I=\cap_{i=1}^l M_i$ and $\sqrt{M_1},\ldots,\sqrt{M_l}$ are distinct. Let $G:=\{M_i: M_i +J \neq A, i\in \{1,\ldots,r\}\} \subset F$. Then $C:=\{\sqrt{M}:M \in G \}$ is an isolated set of prime ideals of $I$. Note that $D:=\{\sqrt{M}:M \in F \} \supset C$ also has the property that if $\sqrt{M_1} \in D$ and $\sqrt{M_2} \subset \sqrt{M_1}$ then $\sqrt{M_2} \in D$. Let $S:=A - \cup_{P \in D} P$. Then for every $i=1,\ldots, l$, by prime avoidance, we have that $S \cap \sqrt{M_i}= \emptyset$ if and only if $M_i \in F$; and for every $i=1,\ldots, r$, we have $S \cap \sqrt{M_i}= \emptyset$ if and only if $M_i \in G$. Consider $\cap_{i=1}^r(S^{-1} M_i)^c=(S^{-1}(\cap_{i=1}^r M_i))^c=(S^{-1}I)^c=(S^{-1}(\cap_{i=1}^l M_i))^c=\cap_{i=1}^l(S^{-1} M_i)^c$. Therefore, $\cap_{M \in G} M = \cap_{M \in F} M = \cap_{Q \in E} Q$, i.e. $\beta(I,J)=\cap_{Q \in E} Q$.
\end{proof}

\begin{proposition} \label{Prop:BasicProperties}
Let $I$, $J$, $H$ and $L$ be ideals in $A$. Then the following statements hold.
\begin{enumerate}
  \item $\beta(I,\sqrt{J})=\beta(I,J)$;
  \item $\beta(\sqrt{I},J)=\sqrt{\beta(I,J)}$;
  \item $\beta(I \cap H,J)=\beta(I,J) \cap \beta(H,J)$;
  \item $\beta(I,J \cap L)=\beta(I,J) \cap \beta(I,L)$;
  \item $\beta(\beta(I,J),J))=\beta(I,J)$;
  \item $\beta(I,\beta(J,I)))=\beta(I,J)$;
  \item $\beta(\beta(I,J),L))=\beta(\beta(I,L),J))$.
\end{enumerate}
\end{proposition}

\begin{proof}
Let $I= \cap_{i=1}^k W_i$ and $J= \cap_{l=1}^r U_l$ be minimal primary decompositions.

(1) For any $W_i$, $W_i+J \neq A$ if and only if $W_i+\sqrt{J} \neq A$. Thus, $\beta(I,\sqrt{J})=\beta(I,J)$.

(2) We also have $W_i+J \neq A$ if and only if $\sqrt{W_i}+J \neq A$ for any $W_i$. From the minimal primary decomposition $I=\cap_{i=1}^k{W_i}$, we can get a minimal primary decomposition $\sqrt{I}= \cap \{ P : \exists i \in \{1,\ldots,k\} \textup{ s.t. } P=\sqrt{W_i} \textup{ and $P$ is a minimal prime belonging to $I$}\}$. Suppose we have already numbered these $W_i$ such that the first $m$ ones satisfy the condition that $W_i + J \neq A$. Then, $\beta(\sqrt{I},J)=\cap \{ P : \exists i \in \{1,\ldots,m\} \textup{ s.t. } P=\sqrt{W_i} \wedge \textup{$P$ minimal}\}$. Therefore, $\sqrt{\beta(I,J)}= \sqrt{\cap_{i=1}^m W_i} = \cap_{i=1}^m \sqrt{W_i}= \cap \{ P : \exists i \in \{1,\ldots,m\} \textup{ s.t. } P=\sqrt{W_i} \wedge \textup{$P$ minimal}\} = \beta(\sqrt{I},J)$. The third equality in the last sentence uses the fact that $\{\sqrt{W_i}:i=1,\dots,m\}$ is an isolated set of prime ideals belonging to $I$ (see the proof of \mbox{Proposition \ref{Prop:WellDefined}}).

(3) Let $H= \cap_{j=1}^s N_j$ is a minimal primary decomposition. Consider the primary decomposition $I \cap H = (\cap_{i=1}^k W_i) \cap (\cap_{j=1}^s N_j)$. It is not necessarily minimal. Let $E_1= \{W_i: W_i+J \neq A \}$ and $E_2= \{N_j: N_j+J \neq A \}$.  By Lemma \ref{Lemma:E_CPart}, we know $\beta(I \cap H, J)= \cap _{Q \in E_1 \cup E_2} Q = (\cap_{Q \in E_1} Q) \cap (\cap_{Q \in E_2} Q) = \beta(I,J) \cap \beta(H,J)$.

(4) $W_i + (J \cap L) \neq A$ if and only if $W_i + J \neq A$ or $W_i +L \neq A$.

(5) It is obvious by Definition \ref{Def:CPart}.

(6) For a fixed $W_i$, if $W_i+U_l = A$ for every $l=1,\ldots,r$, then we have $r$ equalities $u_l=1-w_l$ where $u_l \in U_l$ and $w_l \in W_i$. Thus, $u:=\prod_{l=1}^r u_l = \prod_{l=1} ^r (1-w_l) := 1 - w$. We can see $u \in \cap_{l=1}^r U_l =J$ and $w \in W_i$, and consequently, $W_i + J =A$.  Therefore, $W_i$ is a component of $\beta(I,J)$ if and only if $W_i+J \neq A$, if and only if there exists an $l \in \{1,\ldots,r\}$ such that $W_i+U_l \neq A$, if and only if $(W_i+U_l \neq A) \wedge (U_l + I \neq A)$, if and only if $U_l$ is a component of $\beta(J,I)$ and $W_i+U_l \neq A$, if and only if $W_i + \beta(J,I) \neq A$, if and only if $W_i$ is a component of $\beta(I,\beta(J,I))$. As a result, $\beta(I,\beta(J,I)))=\beta(I,J)$.

(7) $W_i$ is a component of $\beta(\beta(I,J),L))$ if and only if $W_i + J \neq A$ and $W_i +L \neq A$, if and only if $W_i +L \neq A$ and $W_i + J \neq A$, if and only if $W_i$ is a component of $\beta(\beta(I,L),J))$. Thus, $\beta(\beta(I,J),L))=\beta(\beta(I,L),J))$.
\end{proof}

We give another proof of the third equality in the above proposition after the following easy lemma.

\begin{lemma} \label{Lemma:IH}
For any two polynomial ideals $H, I\subset A$ and a control ideal $J \subset A$, we have $(I \cap H )'=I' \cap H'$.
\end{lemma}

\begin{proof}
First, we have $(I \cap H )' \subset I'$ and $(I \cap H )' \subset H'$, and thus $(I \cap H )' \subset I' \cap H'$. Then, for any $\tilde{f} \in I' \cap H'$, as in the proof of Theorem \ref{Th:lift}, we have $\tilde{f}=\tilde{f}|_{t_1=g_1,\ldots,t_{s}=g_{s}} + \sum_{w=1}^{s} r_w(t_w-g_w)$ where $r_w \in A'$. Note that $\tilde{f}|_{t_1=g_1,\ldots,t_{s}=g_{s}} \in I \cap H$. Thus, $\tilde{f} \in (I \cap H)'$ which implies $I' \cap H' \subset (I \cap H)'$. Therefore, $(I \cap H )'=I' \cap H'$ holds.
\end{proof}

The following proof seems easier than the original one, but needs more preparations.

\begin{proof}[Another proof of Proposition \ref{Prop:BasicProperties}(3)]
By Corollary \ref{Corollary:membership} and Lemma \ref{Lemma:IH}, we have $\beta(I \cap H,J)= (I \cap H)'^e \cap A = (I' \cap H')^e \cap A = I'^e \cap H'^e \cap A = (I'^e \cap A) \cap (H'^e \cap A) =\beta(I,J) \cap \beta(H,J)$.
\end{proof}

\bibliographystyle{elsart-harv}
\bibliography{JSC}

\begin{thebibliography}{19}
\expandafter\ifx\csname natexlab\endcsname\relax\def\natexlab#1{#1}\fi
\expandafter\ifx\csname url\endcsname\relax
  \def\url#1{\texttt{#1}}\fi
\expandafter\ifx\csname urlprefix\endcsname\relax\def\urlprefix{URL }\fi

\bibitem[{Alonso et~al.(1990)Alonso, Mora, and Raimondo}]{Alonso91}
Alonso, M.~E., Mora, T., Raimondo, M., August 1990. Local decomposition
  algorithms. In: {Proceedings of AAECC-8 (Applied Algebra, Algebraic
  Algorithms and Error-Correcting Codes)}. Lecture Notes in Computer Science
  508, Springer, Tokyo, Japan, pp. 208--221.

\bibitem[{Atiyah and MacDonald(1969)}]{Atiyah1969}
Atiyah, M.~F., MacDonald, I.~G., 1969. Introduction to Commutative Algebra.
  Addison-Wesley.

\bibitem[{Buchberger(1965)}]{Buchberger65}
Buchberger, B., 1965. Ein Algorithmus zum Auffinden der Basiselemente des
  Restklassenringes nach einem nulldimensionalen Polynomideal. Ph.D. Thesis,
  Mathematical Institute, University of Innsbruck, Austria.

\bibitem[{Buchberger(2006)}]{Buch06}
Buchberger, B., March--April 2006. Bruno {Buchberger's} {PhD} theis 1965: An
  algorithm for finding the basis elements in the residue class ring modulo a
  zero dimensional polynomial ideal. Journal of Symbolic Computation 41~(3--4),
  475--511.

\bibitem[{Cox et~al.(2005)Cox, Little, and O'Shea}]{Cox05}
Cox, D., Little, J., O'Shea, D., 2005. Using Algebraic Geometry (Second
  Edition). Springer, USA.

\bibitem[{Decker et~al.(2012)Decker, Greuel, Pfister, and Sch\"onemann}]{DGPS}
Decker, W., Greuel, G.-M., Pfister, G., Sch\"onemann, H., 2012. {\sc Singular}
  {3-1-6} --- {A} computer algebra system for polynomial computations.
  \url{http://www.singular.uni-kl.de}.

\bibitem[{Eisenbud(1994)}]{Eisenbud94}
Eisenbud, D., 1994. Commutative Algebra with a View Toward Algebraic Geometry.
  Springer.

\bibitem[{Gr\"abe(1993)}]{Graebe93}
Gr\"abe, H.-G., June 1993. Two remarks on independent sets. Journal of
  Algebraic Combinatorics 2~(2), 137--145.

\bibitem[{Gr\"abe(1994)}]{Graebe94}
Gr\"abe, H.-G., December 1994. The tangent cone algorithm and homogenization.
  Journal of Pure and Applied Algebra 97~(3), 303--312.

\bibitem[{Gr\"abe(1995)}]{Graebe95}
Gr\"abe, H.-G., June 1995. Algorithms in local algebra. Journal of Symbolic
  Computation 19~(6), 545--557.

\bibitem[{Greuel and Pfister(1996)}]{Greuel96}
Greuel, G.-M., Pfister, G., February 1996. Advances and improvements in the
  theory of standard bases and syzygies. Archiv der Mathematik 66~(2),
  163--176.

\bibitem[{Greuel and Pfister(2008)}]{Greuel07}
Greuel, G.-M., Pfister, G., 2008. A \textsc{Singular} Introduction to
  Commutative Algebra (Second, Extended Edition). Springer, New York.

\bibitem[{Hironaka(1964)}]{Hironaka64}
Hironaka, H., January--March 1964. Resolution of singularities of an algebraic
  variety over a field of characteristic zero: {I, II}. Annals of Mathematics
  79~(1--2), 109--326.

\bibitem[{Kredel and Weispfenning(1988)}]{Kredel88}
Kredel, H., Weispfenning, V., October 1988. Computing dimension and independent
  sets for polynomial ideals. Journal of Symbolic Computation 6~(2-3),
  231--247.

\bibitem[{Lazard(1983)}]{Lazard83}
Lazard, D., March 1983. Gr\"obner bases, {Gaussian} elimination and resolution
  of systems of algebraic equations. In: {Proceedings of EUROCAL'83 (European
  Computer Algebra Conference London)}. Lecture Notes in Computer Science 162,
  Springer, England, pp. 146--156.

\bibitem[{Liang(2014)}]{Liang14}
Liang, Y., 2014. Solving polynomial equations with equation constraints: the
  zero-dimensional case. \url{arxiv.org/pdf/1408.3639v1.pdf}.

\bibitem[{Mora(1982)}]{Mora82}
Mora, F., April 1982. An algorithm to compute the equations of tangent cones.
  In: {Proceedings of EUROCAM82 (European Computer Algebra Conference
  Marseille)}. Lecture Notes in Computer Science 144, Springer, France, pp.
  158--165.

\bibitem[{Mora and Rossi(1995)}]{Mora95}
Mora, T., Rossi, M.~E., 1995. An algorithm for the {Hilbert-Samuel} function of
  a primary ideal. Communications in Algebra 23~(5), 1899--1911.

\bibitem[{Robbiano(1985)}]{Robbiano1985}
Robbiano, L., April 1985. Term orderings on the polynomial ring. In:
  {Proceedings of EUROCAL'85 (European Conference on Computer Algebra Linz)}.
  Lecture Notes in Computer Science 204, Springer, Austria, pp. 513--517.

\end{thebibliography}

\end{document}